\documentclass[notitlepage,12pt]{article}
\usepackage{amsthm}
\usepackage{amsmath}
\usepackage{amssymb}
\usepackage{latexsym}
\usepackage{mathrsfs}
\usepackage{graphicx}
\usepackage{subfigure}
\usepackage{hyperref}
\hypersetup{colorlinks=true, citecolor=red, urlcolor=blue, linkcolor=blue}
\usepackage{color}
\usepackage{mathtools}
\usepackage{bbm}
\usepackage{float}
\usepackage{indentfirst}
\usepackage{cite}
\usepackage{authblk}
\usepackage{geometry}
\usepackage{cleveref}
\usepackage{times}

\newtheorem{theorem}{Theorem}

\newtheorem{proposition}{Proposition}
\newtheorem{corollary}{Corollary}

\newcommand{\ket}[1]{\left|{#1}\right\rangle}
\newcommand{\bra}[1]{\left\langle{#1}\right|}
\newcommand{\braket}[2]{\langle{#1}|{#2}\rangle}
\newcommand{\brakets}[1]{\langle{#1}\rangle}
\newcommand{\braketd}[3]{\langle{#1}|{#2}|{#3}\rangle}

\newcommand{\ketbra}[2]{\left|{#1}\rangle\!\langle{#2}\right|}
\newcommand{\ketp}[1]{|{#1}\rangle}
\newcommand{\brap}[1]{\langle{#1}|}

\newcommand{\norm}[1]{\vert{#1}\vert}

\begin{document}
\title{\Huge Quantum Uncertainty Equalities and Inequalities for Unitary Operators}
\author[1]{Ao-Xiang Liu}
\author[1]{Ma-Cheng Yang}
\author[1,2]{Cong-Feng Qiao\thanks{\href{mailto:qiaocf@ucas.ac.cn}{qiaocf@ucas.ac.cn}}}
\affil[1]{School of Physical Sciences, University of Chinese Academy of Sciences, Beijing 100049, China}
\affil[2]{Key Laboratory of Vacuum Physics, University of Chinese Academy of Sciences, Beijing 100049, China}
\renewcommand*{\Affilfont}{\small\it} 
\renewcommand*{\Authfont}{\normalsize} 
\renewcommand\Authands{ and } 
\date{Dated: October 25, 2024}
\maketitle
\vspace{-2em}
\begin{abstract}
We explore the uncertainty relation for unitary operators in a new way and find two uncertainty equalities for unitary operators, which are minimized by any pure states. Additionally, we derive two sets of uncertainty inequalities that unveil hierarchical structures within the realm of unitary operator uncertainty. Furthermore, we examine and compare our method for unitary uncertainty relations to other prevailing formulations. We provide explicit examples for better understanding and clarity. Results show that the hierarchical unitary uncertainty relations establish strong bounds. Moreover, we investigate the higher-dimensional limit of the unitary uncertainty equalities.
\end{abstract}
\section{Introduction}
The Heisenberg uncertainty principle \cite{HW27U} lies at the foundation of quantum theory and marks the departure of quantum physics from the classical realm. It was originally formulated for position and momentum observables \cite{HW27U} and sets limits on the information obtainable in simultaneous measurement of conjugate observables and on the preparation of the quantum states giving definite values for incompatible physical realities. Then Robertson \cite{RH29T} extended the uncertainty relation for position and momentum to two arbitrary observables $A$ and $B$ as
\begin{equation}
\Delta A\Delta B\geqslant\frac{1}{2}\norm{\braketd{\psi}{[A,B]}{\psi}},
\end{equation}
where the uncertainties of observables are measured in terms of the standard deviation $\Delta A(B)$ and $\braketd{\psi}{[A,B]}{\psi}$ denotes the average of the commutator of $A$ and $B$ in the state $\ket{\psi}$. Schrödinger strengthened this uncertainty relation by restoring an addtional positive term on the right-hand side of the inequality \cite{SE30S}:
\begin{equation}
\Delta A^2\Delta B^2\geqslant\frac{1}{4}\norm{\brakets{[A,B]}}^2+\frac{1}{4}\norm{\brakets{\lbrace A,B\rbrace}-2\brakets{A}\brakets{B}}^2,
\end{equation}
where $\lbrace A,B\rbrace$ is the anticommutator of $A$ and $B$ with $\brakets{\cdot}$ denoting the expectation value.

However, the above two uncertainty relations become trivial when $\ket{\psi}$ is the eigenstate of $A$ or $B$. To overcome the triviality problem, two stronger uncertainty relations based on the sum of variances of $A$ and $B$ are proposed by Maccone and Pati \cite{ML14S}, though the triviality problem has not been completely removed in their try.
Other variance-based uncertainty relations can be found in Refs. \cite{BP14C,SQ16S,SQ17A,YM23C}. Notably, the uncertainty relations mentioned above are all formulated in inequalities. In contrast, Yao {\it et al.} constructed and formulated two quantum uncertainty \emph{equalities} \cite{YY15I}, which are suitable for any two incompatible observables and imply the uncertainty relations obtained by Maccone and Pati \cite{ML14S}.
Information entropies can also be used to quantify the uncertainty principle \cite{EH57R,HI57A}. The entropic uncertainty relations \cite{BW75I,BB75U,DD83U,MH88G,SJ93E,WS10E,PJ12U,PJ15S,RA16E,PJ17E} have become a fundamental component in the security analysis of nearly all quantum cryptographic protocols, including quantum key distribution and two-party quantum cryptography. Another kind of novel and universal uncertainty relations \cite{PM11M,FS13U,PZ13M,NV16U,LJ19T,LA23A} is based on the majorization theory \cite{MW11I}, which reveals the structure nature of the quantum uncertainty and can reproduce the entropic uncertainty relations by Schur-concave functions. In all, the uncertainty nature of quanta plays a key role in quantum theory and quantum information.

The above mentioned uncertainty relations focus on the operators corresponding to quantum observables. Recently, uncertainty relation for reversible transformation of a quantum system (represented by unitary operators, while not necessarily to be Hermitian), i.e., unitary uncertainty relation, attracts much interest. The formulation of the uncertainty relation for unitary operators, based on variances, closely resembles the approach employed for incompatible observables. In this context, the uncertainty associated with a unitary operator $U$ is quantitatively defined as \cite{MS08U}
\begin{equation}
\label{Eq-var}
\Delta U^2\coloneqq\braketd{\psi}{U^\dagger U}{\psi}-\braketd{\psi}{U^\dagger}{\psi}\braketd{\psi}{U}{\psi}=1-\norm{\braketd{\psi}{U}{\psi}}^2\ .
\end{equation}
The variance-based unitary uncertainty relation is physically meaningful to explore: (i) It may characterize how much a quantum state can be simultaneously localized in two mutually unbiased bases that are related by a discrete Fourier transform \cite{MS08U}. (ii) It constrains our ability to distinguish two distinct unitary evolutions from the original state \cite{BS16U,FG04A,SE05K,AY87P,PA91R,SD17G}. (iii) It is impossible to prepare quantum states that demonstrate maximum interference visibility for two incompatible unitary operators \cite{SE00G}.

Massar and Spindel have demonstrated a unitary uncertainty relation (MSUUR) for any two finite dimensional unitary operators that satisfy the commutation relation $UV=e^{i\phi}VU$. That is
\begin{equation}
\label{Eq-MS}
(1+2K)\Delta U^2\Delta V^2+K^2(\Delta U^2+\Delta V^2)\geqslant K^2
\end{equation}
with $K=\tan\frac{\phi}{2}$, which applies as well to the discrete Fourier transform \cite{MS08U}. Recently, Bagchi and Pati \cite{BS16U} have derived three variance-based unitary uncertainty relations (BPUURs) applicable to general pairs of unitary operators. These relations have been investigated by \cite{XL17E} in the context of a photonic system. Among the three BPUURs, two are particularly stringent, and can be stated as follows:
\begin{equation}
\label{Eq-BPUUR1}
\begin{split}
\normalfont{\text{BPUUR1}}:\quad
\Delta U^2+\Delta V^2\geqslant
&1+\norm{\braketd{\psi}{U^\dagger V}{\psi}}^2\\
&-\braketd{\psi}{U^\dagger V}{\psi}\brakets{U}\brakets{V^\dagger}-\brakets{V}\brakets{U^\dagger}\braketd{\psi}{V^\dagger U}{\psi},
\end{split}
\end{equation}
where the Cauchy-Schwarz inequality is employed in the derivation.
\begin{equation}
\label{Eq-BPUUR2}
\normalfont{\text{BPUUR2}}:\quad\Delta U^2 +\Delta V^2\geqslant
\norm{\braketd{\psi}{U^\dagger\pm iV^\dagger}{\psi^{\perp}}}^2
\mp 2\,\mathrm{Im}[\mathrm{Cov}(U,V)].
\end{equation}
Here $\mathrm{Im}[\mathrm{Cov}(U,V)]$ represents the imaginary part of $\mathrm{Cov}(U,V)\coloneqq\braketd{\psi}{U^\dagger V}{\psi}-\brakets{U^\dagger}\brakets{V}$ and $|\psi^\perp\rangle$ symbolizes the arbitrary state that is orthogonal with $\ket{\psi}$. More recently,
Bong {\it et al.} \cite{BK18S} have presented a unitary uncertainty relation (BUUR) that applies to any number of unitary operators. In the case of two unitary operators, a well-known BUUR is formulated as follows \cite{BK18S}:
\begin{equation}
\label{Eq-BUUR}
\Delta U^2\Delta V^2\geqslant\norm{\mathrm{Cov}(U,V)}^2.
\end{equation}
A few other results of unitary uncertainty relations can be found in Refs. \cite{MM12T,MM13T,YB19S, HX22I,QD21E,SJ21E}.

In this paper, we introduce two novel variance-based \emph{unitary uncertainty equalities} (UUEs) in the form of sums and products. These UUEs are proven to hold for \emph{any} pure states, thereby establishing a comprehensive understanding of unitary uncertainty. By relaxing the conditions imposed in obtaining the UUEs, we find two hierarchical frameworks of unitary uncertainty relations. Remarkably, our unitary uncertainty relations encompass the BPUUR2 as a specific case, and the certain unitary uncertainty relations found in this study surpass the MSUUR, BPUURs, and BUUR in tightness of the lower bound. Moreover, we show that, with the dimensionality ($d$) of the system approaching infinity, the unitary uncertainty equalities and inequalities for the discrete unitary operators related by the discrete Fourier transform reduce to the uncertainty equalities \cite{YY15I} and inequalities \cite{ML14S} of Hermitian operators, respectively.

This paper is structured as follows: Section \ref{Sec-NUE} presents the derivation of uncertainty equalities for non-Hermitian operators. In Section \ref{Sec-UUE}, we introduce the uncertainty equalities and two specific sets of uncertainty inequalities applicable to unitary operators. A comparative analysis of the lower bounds of different unitary uncertainty relations with explicit examples will be presented. We discuss the uncertainty equalities and inequalities for mixed states in Sec. \ref{Sec-UUEMixed}. The investigation of the limits of unitary uncertainty equalities in higher dimensions is explored in Section \ref{Sec-HDL}. Finally, Section \ref{Sec-Con} provides some concluding remarks.
\section{Uncertainty Equalities}
\label{Sec-NUE}
In this section, we extend those two \emph{uncertainty equalities} for Hermitian operators in Ref.\cite{YY15I} to general quantum operators (Hermitian or non-Hermitian).
\begin{proposition}
\label{Thm-UE1}
Given any pair of general operators $A$ and $B$, the following uncertainty equality holds
\begin{equation}
\label{Eq-UE1}
\Delta A^2 +\Delta B^2=
\sum_{k=1}^{d-1}\norm{\braketd{\psi}{A^\dagger\pm iB^\dagger}{\psi_k^{\perp}}}^2
\mp2\,\mathrm{Im}[\mathrm{Cov}(A,B)],
\end{equation}
where the variance of general operator $O$ is defined as $\Delta O^2\coloneqq\braketd{\psi}{O^\dagger O}{\psi}-\norm{\braketd{\psi}{O}{\psi}}^2$, the covariance of $A$, $B$ is defined by $\mathrm{Cov}(A,B)=\braketd{\psi}{A^\dagger B}{\psi}-\braketd{\psi}{A^\dagger}{\psi}\braketd{\psi}{B}{\psi}$ and $\lbrace\ket{\psi},|\psi^{\perp}_k\rangle|_{k=1,...,d-1} \rbrace$ constitutes an orthonormal complete basis in the $d$-dimensional Hilbert space.
\end{proposition}
\begin{proof}
Given an operator $\Psi^\perp=\mathbbm{1}-\ketbra{\psi}{\psi}$ and a state $\ket{f^\pm}=(A\pm iB)\ket{\psi}$, we have
\begin{equation}
\label{Prf-UE1}
\begin{split}
\braketd{f^\mp}{\Psi^\perp}{f^\mp}
&=\bra{\psi}(A^\dagger\pm iB^\dagger)(\mathbbm{1}-\ketbra{\psi}{\psi})(A\mp iB)\ket{\psi}\\
&=\braket{f^\mp}{f^\mp}-\braket{f^\mp}{\psi}\braket{\psi}{f^\mp}\\
&=\brakets{A^\dagger A+B^\dagger B\mp iA^\dagger B\pm iB^\dagger A}-(\brakets{A^\dagger}\pm i\brakets{B^\dagger})(\brakets{A}\mp i\brakets{B})\\
&=\Delta A^2 +\Delta B^2\pm 2\,\mathrm{Im}[\mathrm{Cov}(A,B)]\ .
\end{split}
\end{equation}
Notice that $\lbrace\ket{\psi},|\psi^{\perp}_k\rangle|_{k=1,...,d-1} \rbrace$  constitutes an orthonormal complete basis, the operator $\Psi^\perp$ can be arbitrarily decomposed as $\Psi^\perp=\sum_{k=1}^{d-1}\ketbra{\psi_k^\perp}{\psi_k^\perp}$. The uncertainty equality in Eq. (\ref{Eq-UE1}) follows by substituting the $\Psi^\perp$ in Eq. (\ref{Prf-UE1}) with the decomposition.
\end{proof}
\begin{proposition}
\label{Thm-UE2}
Given any pair of general operators $A$ and $B$, the following uncertainty equality holds (considering only nontrivial case, i.e., $\Delta A\Delta B>0$)
\begin{equation}
\label{Eq-UE2}
\Delta A\Delta B=
\sum_{k=1}^{d-1}\frac{\norm{\braketd{\psi}{A^\dagger\Delta B\pm iB^\dagger\Delta A}{\psi_k^{\perp}}}^2}{2\Delta A\Delta B}
\mp\,\mathrm{Im}[\mathrm{Cov}(A,B)]\ .
\end{equation}
\end{proposition}
\noindent
To prove this uncertainty equality, one can start by defining state $\ket{g^\pm}=(\frac{A}{\Delta A}\pm i\frac{B}{\Delta B})$ and follow the procedure in the proof of Proposition \ref{Thm-UE1}. Note, the Prop. \ref{Thm-UE1} and Prop. \ref{Thm-UE2} reduce to the uncertainty equalities for Hermitian operators obtained by Yao {\it et al.} \cite{YY15I}.
\section{Unitary Uncertainty Equalities and Inequalities}
\label{Sec-UUE}
In this section, we consider the uncertainty equalities and relations (inequalities) for two arbitrary unitary operators.
\subsection{Unitary uncertainty equalities}

One can use the Prop. \ref{Thm-UE1} and Prop. \ref{Thm-UE2} to obtain the sum-form unitary uncertainty equality (UUES) and product-form unitary uncertainty equality (UUEP) for two arbitrary unitary operators, respectively, by substituting $A$ and $B$ with unitary operators $U$ and $V$.
\begin{theorem}
Given any two unitary operators $U$ and $V$, the following unitary uncertainty equalities exist
\begin{equation}
\label{Eq-UUES}
\hspace{-1em}\normalfont{\text{UUES}}:\quad
\Delta U^2 +\Delta V^2=
\sum_{k=1}^{d-1}\norm{\braketd{\psi}{U^\dagger\pm iV^\dagger}{\psi_k^{\perp}}}^2
\mp 2\,\mathrm{Im}[\mathrm{Cov}(U,V)],
\end{equation}
\begin{equation}
\label{Eq-UUEP}
\normalfont{\text{UUEP}}:\quad
\Delta U\Delta V=
\sum_{k=1}^{d-1}\frac{\norm{\braketd{\psi}{U^\dagger\Delta V\pm iV^\dagger\Delta U}{\psi_k^{\perp}}}^2}{2\Delta U\Delta V}
\mp\,\mathrm{Im}[\mathrm{Cov}(U,V)].
\end{equation}
where $\Delta U^2$, $\Delta V^2$ are defined in Eq. (\ref{Eq-var}), $\mathrm{Cov}(U,V)=\braketd{\psi}{U^\dagger V}{\psi}-\braketd{\psi}{U^\dagger}{\psi}\braketd{\psi}{V}{\psi}$ and $\lbrace\ket{\psi},|\psi^{\perp}_k\rangle|_{k=1,...,d-1} \rbrace$  constitutes an orthonormal complete basis in the $d$-dimensional Hilbert space.
\end{theorem}

In the following we provide some observations on the aforementioned unitary uncertainty equalities prior to proceeding further. First, it is worth noting that these equalities can be readily extended to infinite-dimensional systems. Secondly, it is important to highlight that the unitary uncertainty equalities, namely UUES and UUEP, are always satisfied by pure states. Moreover, by dropping off $1\leqslant n\leqslant d-2$ terms in the summations of the right-hand sides of Eqs. (\ref{Eq-UUES},\ref{Eq-UUEP}), one may readily obtain two hierarchical structures of unitary uncertainty relations (inequalities). For instance, in Eq. (\ref{Eq-UUES}), discarding $d-2$ perpendicular terms on the right-hand side of Eq. (\ref{Eq-UUES}), we reproduce BPUUR2 of Eq. (\ref{Eq-BPUUR2}), no need to resort to the Cauchy-Schwarz inequality. Similarly, Eq. (\ref{Eq-UUEP}) leads to a new unitary uncertainty relation
\begin{equation}
\Delta U\Delta V\geqslant
\frac{\norm{\braketd{\psi}{U^\dagger\Delta V\pm iV^\dagger\Delta U}{\psi^{\perp}}}^2}{2\Delta U\Delta V} \mp\mathrm{Im}[\mathrm{Cov}(U,V)]\ .
\end{equation}
It is noteworthy that the perpendicular terms on the right-hand sides of Eqs. (\ref{Eq-UUES},\ref{Eq-UUEP}) originate from the fact that an operator maps a state onto a linear combination of itself and other states perpendicular to it, see Fig. \ref{Fig-evolution} for illustration.
\begin{figure}[H]
\centering
\includegraphics[width=0.7\linewidth]{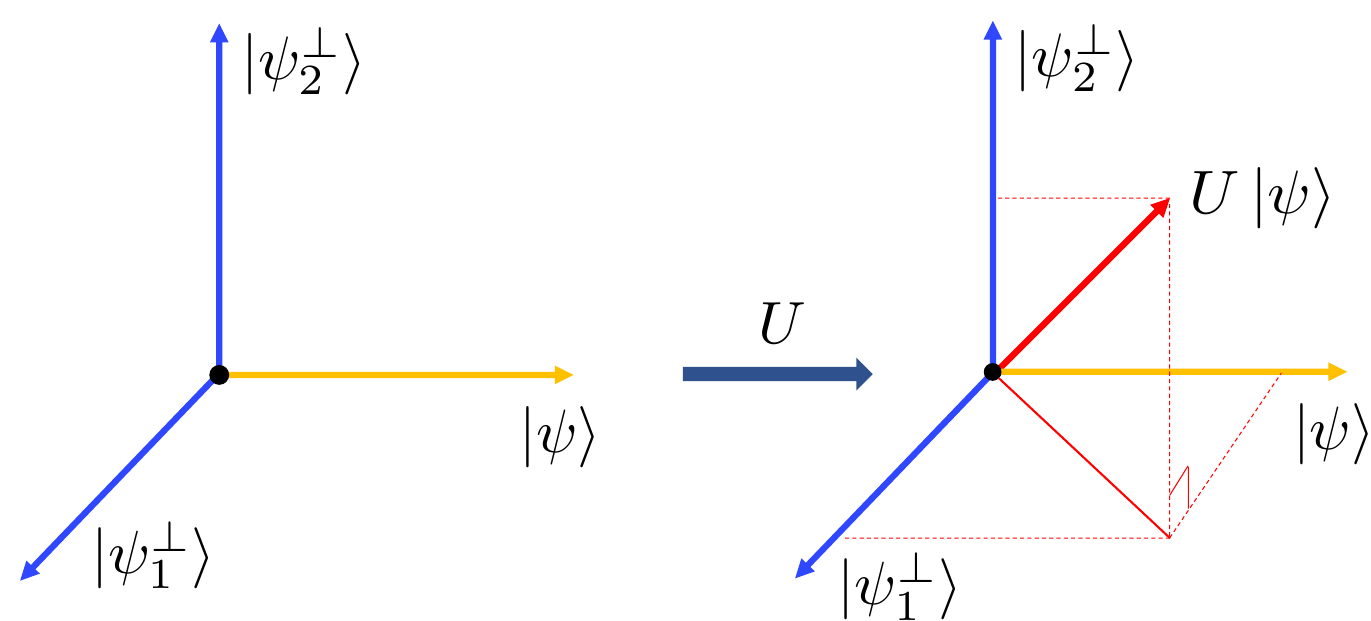}
\caption{In three-dimensional Hilbert space, the unitary operator $U$ maps a state $\ket{\psi}$ onto a linear combination of itself and other perpendicular states. }
\label{Fig-evolution}
\end{figure}
\subsection{The hierarchical structure of unitary uncertainty relations}
\label{Sec-HierarUUR}
As mentioned in preceding section, if we consider only $1\leqslant n < d-1$ elments in the set of $\lbrace|\psi^{\perp}_k\rangle|_{k=1,...,d-1} \rbrace$, we can reach a series of unitary uncertainty relations from UUES and UUEP. Taking UUES as an example, it is straightforward to find
\begin{equation}
\Delta U^2 +\Delta V^2 >
\sum_{k\in\mathcal{I}_n}\norm{\braketd{\psi}{U^\dagger\pm iV^\dagger}{\psi_k^{\perp}}}^2
\mp 2\,\mathrm{Im}[\mathrm{Cov}(U,V)]\ ,
\end{equation}
where $\mathcal{I}_n\subseteq\{1,\cdots,d-1\}$, and $|\mathcal{I}_n|_c=n$ denotes the cardinality of the set $\mathcal{I}_n$. If and only if $n=d-1$, the inequality turns to equality. We denote the set of all the possible $\mathcal{I}_n$ as
\begin{equation}
S_n=\big\lbrace \mathcal{I}_n\subseteq\{1,\cdots,d-1\}\, \big| \,|\mathcal{I}_n|_c=n\big\rbrace\ .
\end{equation}
Here, the cardinality of $S_n$, the $|S_n|_c$, is $\frac{(d-1)!}{n!(d-n-1)!}$.
Furthermore, by maximizing the term $\sum_{\mathcal{I}_n}\norm{\braketd{\psi}{U^\dagger\pm iV^\dagger}{\psi_k^{\perp}}}^2$ over all possible sets $\mathcal{I}_n$, there will come a stronger unitary certainty and uncertainty relation from UUES.
In the same vein, one can also derive a hierarchical framework of unitary uncertainty relations for UUEP. As a result, we have the  following findings:
\begin{corollary}
Given any two unitary operators $U$ and $V$, there exists the following hierarchical structure unitary certainty and uncertainty relations for every $1\leqslant n\leqslant d-1$
\begin{equation}
\label{Eq-UURS}
\hspace{-1.7em}\normalfont{\text{UURS}}n:\quad
\Delta U^2 +\Delta V^2\geqslant
\max\limits_{\mathcal{I}_n\in S_n}\Big[\sum_{k\in\mathcal{I}_n}\norm{\braketd{\psi}{U^\dagger\pm iV^\dagger}{\psi_k^{\perp}}}^2\Big]
\mp 2\,\mathrm{Im}[\mathrm{Cov}(U,V)]\ .
\end{equation}
\begin{equation}
\label{Eq-UURP}
\normalfont{\text{UURP}}n:\quad
\Delta U\Delta V\geqslant
\frac{1}{2}	\max\limits_{\mathcal{I}_n\in S_n}\Big[\sum_{k\in\mathcal{I}_n}\frac{\norm{\braketd{\psi}{U^\dagger\Delta V\pm iV^\dagger\Delta U}{\psi_k^{\perp}}}^2}{\Delta U\Delta V}\Big] \mp\mathrm{Im}[\mathrm{Cov}(U,V)]\ .
\end{equation}
In this context, the notation $\max\limits_{\mathcal{I}_n\in S_n}[\cdot]$ signifies the process of maximizing the corresponding term by considering all possible sets $\lbrace|\psi^{\perp}_k\rangle|_{k\in\mathcal{I}_n}\rbrace$.
\end{corollary}
\begin{figure}[H]
\centering
\includegraphics[width=1\linewidth]{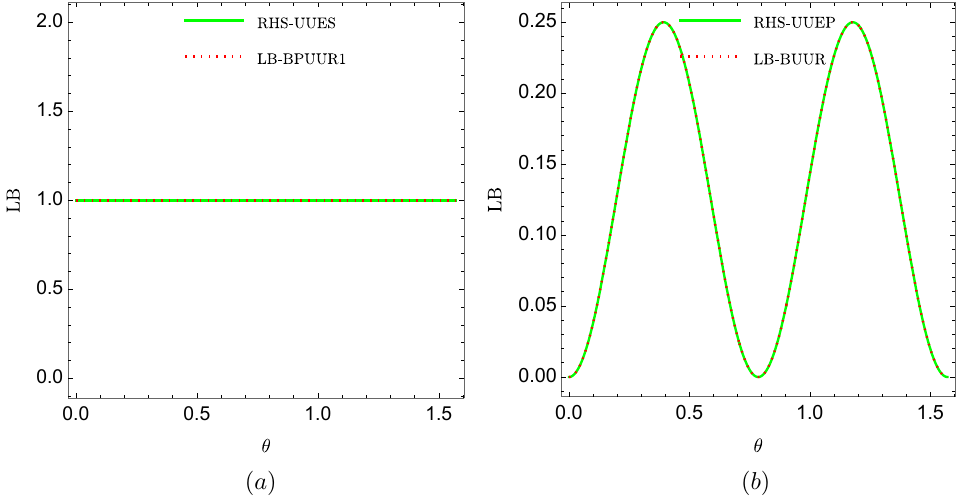}
\vspace{-2em}
\caption{Comparison of the right-hand sides of UUES and UUEP with the lower bounds of BPUUR1 in Eq. (\ref{Eq-BPUUR1}) and BUUR in Eq. (\ref{Eq-BUUR}), respectively, in a qubit system. (a) The dotted red curve represents the lower bound of BPUUR1 (LB-BPUUR1) and coincides with the solid green curve representing the right-hand side of UUES (RHS-UUES). (b) The dotted red curve represents the lower bound of BUUR (LB-BUUR) and coincides with the solid green curve representing the right-hand side of UUEP (RHS-UUEP).}
\label{Fig-UUEqubit}
\end{figure}

Investigating the effectiveness of new unitary uncertainty relations in the context of qubit systems is of significant interest. It is worth noting that for any given state $\ket{\psi}$ of a qubit system, the perpendicular state $|\psi^\perp\rangle$ is uniquely determined, up to an irrelevant overall phase factor. As a result, the Bagchi and Pati's  BPUUR2 in Eq. (\ref{Eq-BPUUR2}), along with UURS1 and UURP1 in Eqs. (\ref{Eq-UURS}, \ref{Eq-UURP}), becomes an equality when considering arbitrary qubit pure states, as shown in Fig.{\ref{Fig-UUEqubit}}. Therefore, the UURS, BPUUR2 and UURP encompass all the uncertainties arising from the sum and product of variances in the case of qubit systems.

In a nutshell, we have presented two hierarchical structures of unitary uncertainty relations for both the sum and product forms, utilizing the variance of the unitary operators. These uncertainty relations become stronger as they address a greater number of perpendicular states $\ket{\psi^\perp_k}$. In other words, the lower bound of UURS(P)$n$ is stronger than UURS(P)$m$ if $n\geqslant m$. It is worth to note that these unitary uncertainty relations are obtained without employing the Cauchy-Schwarz inequality.

\subsection{Examples}
\label{Sec-eg}
For illustration of our results in comparison with those in the literature, let's consider a generic pure state $\ket{\psi}=\cos\theta\ket{0}-\sin\theta\ket{d-1}$ with $\theta\in[0,\pi/2]$, where $\ket{0}= (1,0,\cdots,0)$ and $\ket{d-1}= (0,\cdots,0,1)$ represent two base states in a $d$-dimensional Hilbert space, and two unitary operators that are connected via discrete Fourier transform, for instance
\begin{equation}
\label{Eq-UOU}
U=\mathrm{diag}(1,\omega,\omega^2,...,\omega^{d-1}),
\end{equation}
\begin{equation}
\label{Eq-UOV}
V=
\begin{pmatrix}
0&1\\
I_{d-1}&0
\end{pmatrix}\ .
\end{equation}
Here, $I_{d-1}$ denotes the $(d-1)$-dimensional identity matrix and $\omega=e^{i2\pi/d}$. Obviously, $UV=\omega VU$. Next, we compare our lower bounds with that of Massar and Spindel's, Bagchi and Pati's, and Bong {\it et al.}'s  results for two-, three-, four-, five- and six-dimensional systems, see Fig. \ref{Fig-UUEqubit}, Fig. \ref{Fig-UURS} and Fig. \ref{Fig-UURP} for illustration.

For a qubit system, i.e., $d=2$,
\begin{equation}
\label{Eq-UVQubit}
U=\sigma_z=
\begin{pmatrix}
1&0\\
0&-1
\end{pmatrix},
V=\sigma_x=
\begin{pmatrix}
0&1\\
1&0
\end{pmatrix}.
\end{equation}
In this case, both the lower bounds of unitary uncertainty relations BPUUR1 in (\ref{Eq-BPUUR1}) and BUUR in (\ref{Eq-BUUR}) are maximized as shown by Fig. \ref{Fig-UUEqubit}. Moreover, the sum of the variance of $U$ and $V$ in Eq. (\ref{Eq-UVQubit}) equals to one, i.e.,
\begin{equation}
\Delta U^2+\Delta V^2=1\ ,
\end{equation}
which maximizes the $2$-dimensional MSUUR \cite{MS08U}
\begin{equation}
\label{Eq-URPauli}
\Delta\sigma_z^2+\Delta\sigma_x^2\geqslant 1\ .
\end{equation}
An interpretation of Eq. (\ref{Eq-URPauli}) is that it reveals an uncertainty relation for Mach-Zehnder interferometers, wherein the predictability of the particle's path and the visibility of the interference fringes are interconnected \cite{GD88S,JG95T,MS08U}. By considering the definition of interference visibility as $\mathcal{V}\coloneqq\norm{\braketd{\psi}{U}{\psi}}=\sqrt{1-\Delta U^2}$ \cite{SE00G}, we observe a strict complementarity between the interference visibilities of the unitary operators $U=\sigma_z$ and $V=\sigma_x$ in state $\ket{\psi}=\cos\theta\ket{0}-\sin\theta\ket{1}$:
\begin{equation}
\mathcal{V}^2_U+\mathcal{V}^2_V=1,
\end{equation}
which has operational significance and can be verified experimentally.

\begin{figure}[H]
\centering
\includegraphics[width=1\linewidth]{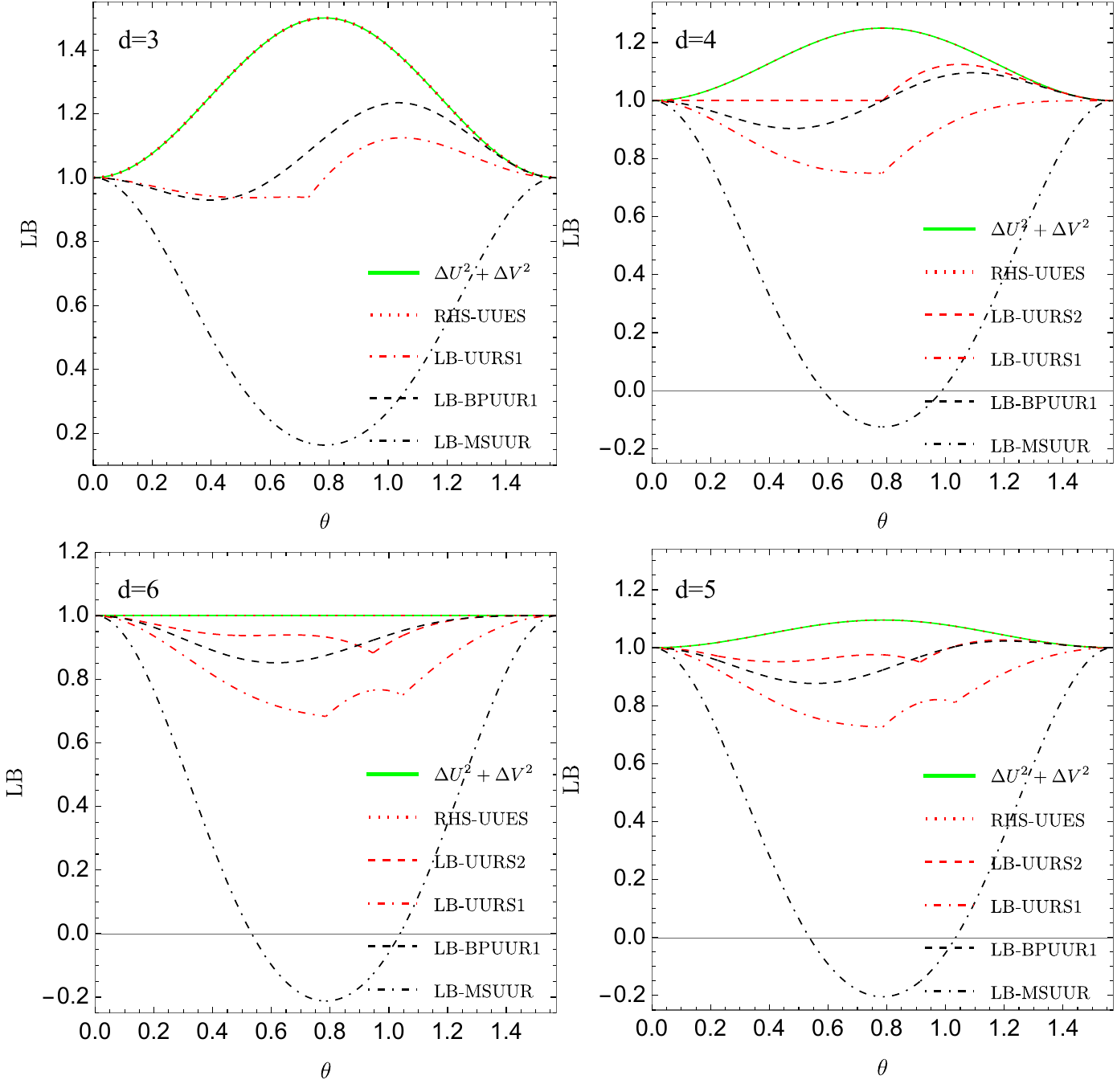}
\vspace{-2em}
\caption{Comparison of the sum of variances, $\Delta U^2+\Delta V^2$ (represented by the solid green curve), and the lower bounds for some aforementioned unitary uncertainty relations. The right-hand side of the UUES (RHS-UUES) is represented by the dotted red curve. The dot-dashed black curve represents the Massar and Spindel's lower bound (LB-MSUUR). The dashed black curve represents the lower bound of Bagchi and Pati's BPUUR1 (LB-BPUUR1). The dot-dashed red curve and dashed red curve represent the lower bound of UURS1 (LB-UURS1) and UURS2 (LB-UURS2), respectively.}
\label{Fig-UURS}
\end{figure}
For $d=3$,
\begin{equation}
U=\mathrm{diag}(1,e^{\frac{2\pi i}{3}},e^{\frac{4\pi i}{3}}),\
V=
\begin{pmatrix}
0&1\\
I_{2}&0
\end{pmatrix}.
\end{equation}
In this case, the $K$ in Eq. (\ref{Eq-MS}) turns to be $K=\tan\frac{\pi}{3}$.
\begin{figure}[H]
\centering
\includegraphics[width=1\linewidth]{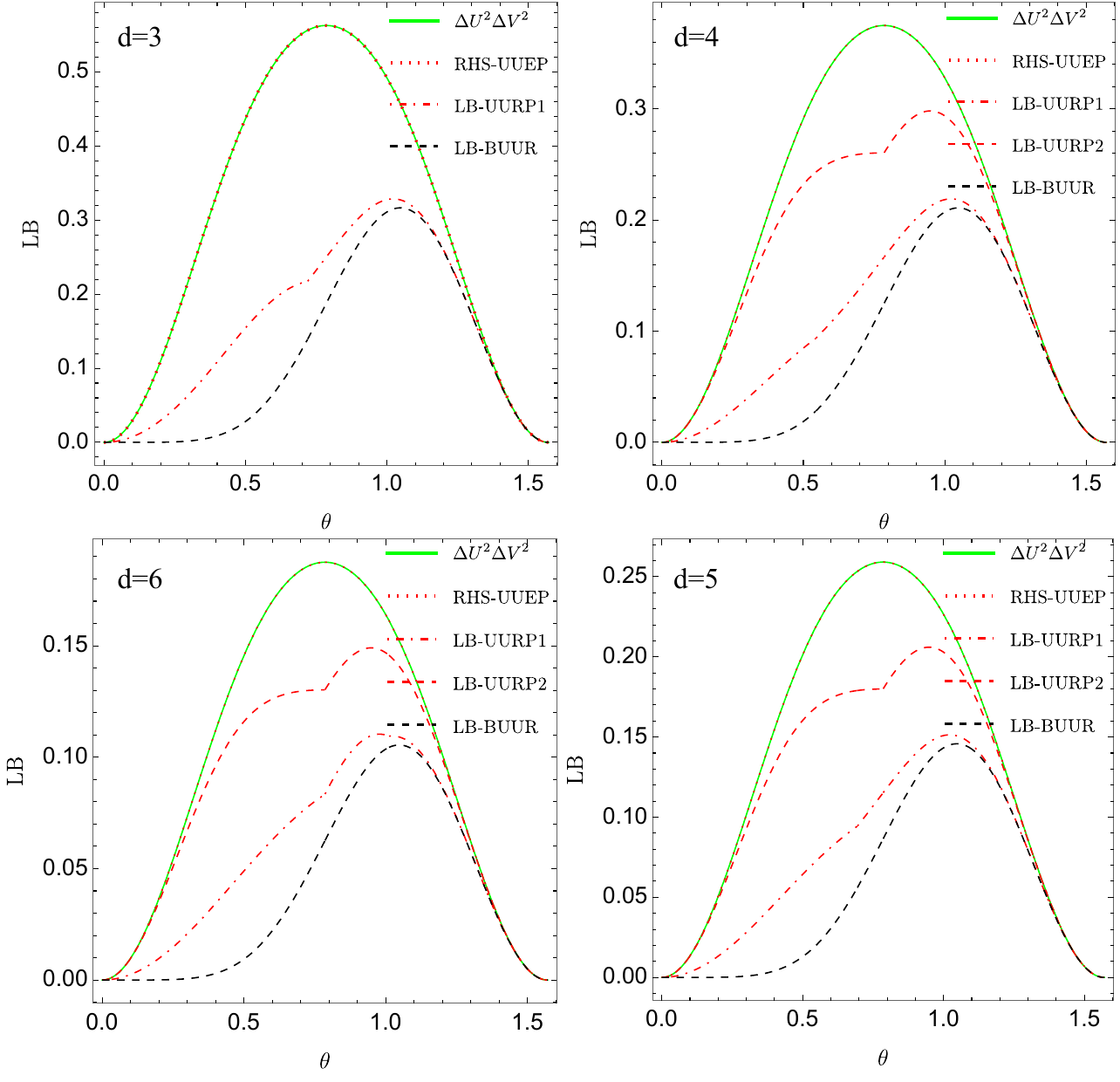}
\vspace{-2em}
\caption{Comparison of the product of variances, $\Delta U^2\Delta V^2$ (represented by the solid green curve), and the lower bounds for some aforementioned unitary uncertainty relations. The right-hand side of the UUEP (RHS-UUEP) is represented by the dotted red curve. The dashed black curve represents the lower bound of Bong {\it et al.}’s BUUR (LB-BUUR). The dot-dashed red curve and dashed red curve represent the lower bound of UURP1 (LB-UURP1) and UURP2 (LB-UURP2), respectively.}
\label{Fig-UURP}
\end{figure}

For $d=4$,
\begin{equation}
U=\mathrm{diag}(1,e^{\frac{\pi i}{2}},e^{\pi i},e^{\frac{\pi i}{2}}),\
V=
\begin{pmatrix}
0&1\\
I_{3}&0
\end{pmatrix},\
K=\tan\frac{\pi}{4}.
\end{equation}

For $d=5$,
\begin{equation}
U=\mathrm{diag}(1,e^{\frac{2\pi i}{5}},e^{\frac{4\pi i}{5}},e^{\frac{6\pi i}{5}},e^{\frac{8\pi i}{5}}),\
V=
\begin{pmatrix}
0&1\\
I_{4}&0
\end{pmatrix},\
K=\tan\frac{\pi}{5}.
\end{equation}

For $d=6$,
\begin{equation}
U=\mathrm{diag}(1,e^{\frac{\pi i}{3}},e^{\frac{2\pi i}{3}},e^{\pi i},e^{\frac{4\pi i}{3}},e^{\frac{5\pi i}{3}}),\
V=
\begin{pmatrix}
0&1\\
I_{5}&0
\end{pmatrix},\
K=\tan\frac{\pi}{6}.
\end{equation}

Fig. \ref{Fig-UURS} shows that the bound of UURS1 is stronger than the Massar and Spindel bound but weaker than Bagchi and Pati bound. It also shows that the bound of UURS2 performs better than the Massar and Spindel bound, and better than the Bagchi and Pati bound for certain regions. Moreover, the curve of the right-hand side of unitary uncertainty equality in Eq. (\ref{Eq-UUES}) always coincides with the curve of $\Delta U^2+\Delta V^2$. In Fig. \ref{Fig-UURP}, it shows that both the bounds of UURP1 and UURP2 are stronger than the bound of unitary uncertainty relation of Bong {\it et al.} Again, the curve of the square of the right-hand side of unitary uncertainty equality in Eq. (\ref{Eq-UUEP}) always coincides with the curve of $\Delta U^2\Delta V^2$. Note that there are only three nonzero terms in the summations of the right-hand sides of UUES and UUEP for both the five- and six- dimensional cases within our choice of the unitary operators and orthogonal basis.
\section{Unitary Uncertainty Equalities and Inequalities for Mixed States}
\label{Sec-UUEMixed}
Up to now, we have discussed unitary uncertainty equalities and inequalities within the case of pure states. In this section, we shall show that our results can be generalized to the case of mixed states utlizing the concept of the {\it purification} of density operator. 

Let $\rho$ be a density operator on the concerned Hilbert space $\mathcal{H}_S$, which represents a mixed state. Then one can always find a pure state vector of an extended Hilbert space, i.e., $\ket{\xi}_{SE}\in\mathcal{H}_S\otimes\mathcal{H}_E$, such that 
\begin{equation}
\rho=\mathrm{Tr}_E(\ketbra{\xi}{\xi}_{SE})
\end{equation}
where $\mathrm{Tr}_E$ stands for the partial trace over the ``environment" Hilbert space $\mathcal{H}_E$. The pure state vector $\ket{\xi}_{SE}$ is a purification of density operator $\rho$. Then one can apply the extended operators $U\otimes\mathbbm{1}_E$ to obtain the results about the density operator $\rho$, where $U$ is any unitary operator on Hilbert space $\mathcal{H}_S$ and $\mathbbm{1}_E$ is identity operator on the Hilbert space $\mathcal{H}_E$. Furthermore, the expectation of $U$ on $\rho$ satisfies the following equation \cite{MA10Q}
\begin{equation}
\mathrm{Tr}(U\otimes\mathbbm{1}_E\ketbra{\xi}{\xi}_{SE})=\mathrm{Tr}(\rho U).
\end{equation}
In accordance,  the variance and covariance can be obtained from the purification
\begin{equation}
\Delta_{S} U=\Delta_{SE} (U\otimes\mathbbm{1}_E), 
\end{equation}
\begin{equation}
\mathrm{Cov}_{S}(U,V)=\mathrm{Cov}_{SE}(U\otimes\mathbbm{1}_E,V\otimes\mathbbm{1}_E).
\end{equation}
Here, the subscripts $S$ and $SE$ indicate that the variance (covariance) is calculated based on the state $\rho$ and $\ket{\xi}_{SE}$, respectively.

Hence, consider two extended operators $U\otimes\mathbbm{1}_E$ and $V\otimes\mathbbm{1}_E$ and a purification $\ket{\xi}_{SE}$ of mixed state $\rho$, one may have the following unitary uncertainty equalities
\begin{equation}
\label{Eq-UUES_mixed}
\Delta U^2 +\Delta V^2=
\sum_{k=1}^{d_s d_e -1}\norm{\mathrm{Tr}\big((U^\dagger\pm iV^\dagger)M_k^\xi\big)}^2
\mp 2\,\mathrm{Im}[\mathrm{Cov}(U,V)],   
\end{equation}
\begin{equation}
\label{Eq-UUEP_mixed}
\Delta U\Delta V=
\sum_{k=1}^{d-1}\frac{|\mathrm{Tr}\big((U^\dagger\Delta V\pm iV^\dagger\Delta U)M_k^\xi\big)|^2}{2\Delta U\Delta V}
\mp\,\mathrm{Im}[\mathrm{Cov}(U,V)].
\end{equation}
Here, $d_s$ and $d_e$ denote the dimensions of Hilbert space $\mathcal{H}_S$ and $\mathcal{H}_E$, respectively. $M_k^\xi=\mathrm{Tr}_E (\ketp{\xi}\brap{\xi_k^\perp}_{SE})$, where $\lbrace\ket{\xi}_{SE},|\xi^{\perp}_k\rangle_{SE}|_{k=1,...,d_s d_e -1} \rbrace$  constitutes an orthonormal complete basis in the Hilbert space $\mathcal{H}_S\otimes\mathcal{H}_E$. Using the approach proposed in section \ref{Sec-HierarUUR}, one can obtain the hierarchical structures of unitary uncertainty inequalities for mixed state $\rho$ accordingly. It is worth mentioning that, in a recent work \cite{LM23U}, the concept of {\it amplitude operator} has been used to generalize the Aharonov-Vaidman identity \cite{AY90P} to more general density operators. Moreover, the equivalence between amplitude operators and purifications has been established as well \cite{LM23U}.

\section{High-dimensional Limit of UUEs}
\label{Sec-HDL}
In Section \ref{Sec-NUE}, we have derived two uncertainty equalities tailored for non-Hermitian operators. It should be mentioned that when confining the operators $A$ and $B$ to be Hermitian, one can replicate the uncertainty equalities initially given in Ref. \cite{YY15I}. On the other hand, reference \cite{MS08U} has demonstrated the existence of a correlation between uncertainty relations about unitary operators and their corresponding Hermitian counterparts. In higher-dimensional limit, i.e., $d\rightarrow\infty$, the Massar and Spindel unitary uncertainty relation reproduces the Heisenberg uncertainty relations for two Hermitian operators. Inspired by this discovery, Ref. \cite{BS16U} embarks on an exploration of high-dimensional limits of the uncertainty relations for two unitary operators related by the discrete Fourier transform. In this section, we explore the high-dimensional limit of our unitary uncertainty equalities for two unitary operators related by the discrete Fourier transform.

First of all, one can rewrite the two unitary operators $U$ and $V$ related by the discrete Fourier transform as \cite{MS08U}
\begin{equation}
U=e^{iu\sqrt{2\pi/d}},\ V=e^{iv\sqrt{2\pi/d}},
\end{equation}
where $u$ and $v$ are Hermitian operators. We consider only a \emph{subset} of quantum states for which one can approximate $U$ and $V$ by their series expansions \cite{MS08U}:
\begin{equation}
U\simeq \mathbbm{1}+i\sqrt{\frac{2\pi}{d}}u-\frac{\pi}{d}u^2,\ V\simeq
\mathbbm{1}+i\sqrt{\frac{2\pi}{d}}v-\frac{\pi}{d}v^2.
\end{equation}
Accordingly, it implies the variances of the left-hand side of unitary uncertainty equalities are
\begin{equation}
\Delta U^2\simeq\frac{2\pi}{d}(\brakets{u^2}-\brakets{u}^2)=
\frac{2\pi}{d}\Delta u^2,\
\Delta V^2\simeq\frac{2\pi}{d}(\brakets{v^2}-\brakets{v}^2)=
\frac{2\pi}{d}\Delta v^2.
\end{equation}
The relevant terms on the right-hand side of the UUES in Eq. (\ref{Eq-UUES}) are
\begin{equation}
\norm{\braketd{\psi_k^{\perp}}{U\mp iV}{\psi}}^2\simeq\norm{\braketd{\psi_k^{\perp}}{(\mathbbm{1}+i\sqrt{\frac{2\pi}{d}}u-\frac{\pi}{d}u^2)\mp i(\mathbbm{1}+i\sqrt{\frac{2\pi}{d}}v-\frac{\pi}{d}v^2)}{\psi}}^2.
\end{equation}
Keeping terms up to the second order, we have
\begin{equation}
\norm{\braketd{\psi_k^{\perp}}{U\mp iV}{\psi}}^2\simeq\frac{2\pi}{d}\norm{\braketd{\psi_k^{\perp}}{u\mp iv}{\psi}}^2.
\end{equation}
Similarly, we obtain the approximate form of $\mathrm{Cov}(U,V)$
\begin{equation}
\mathrm{Im}[\mathrm{Cov}(U,V)]\simeq\frac{\pi}{id}\brakets{[u,v]},
\end{equation}
and
\begin{equation}
\norm{\braketd{\psi_k^{\perp}}{\frac{U}{\Delta U}\mp i\frac{V}{\Delta V}}{\psi}}^2\simeq\frac{2\pi}{d}\norm{\braketd{\psi_k^{\perp}}{\frac{u}{\Delta u}\mp i\frac{v}{\Delta v}}{\psi}}^2,
\end{equation}

Combining the results above, we reproduce the uncertainty equalities for the Hermitian operators $u$ and $v$ as follows
\begin{equation}
\Delta u^2+\Delta v^2=\sum_{k=1}^{d-1}\norm{\braketd{\psi_k^{\perp}}{u\mp iv}{\psi}}^2\pm i\brakets{[u,v]},
\end{equation}
\begin{equation}
\Delta u\Delta v=
\frac{\pm\frac{1}{2}i\brakets{[u,v]}}{1-\frac{1}{2}\sum_{k=1}^{d-1}\norm{\braketd{\psi_k^{\perp}}{\frac{u}{\Delta u}\mp i\frac{v}{\Delta v}}{\psi}}^2},
\end{equation}
which had been obtained in Ref. \cite{YY15I}. Again, if we intentionally keep only one particular $\ket{\psi^\perp}$ term while discarding others, the ensuing uncertainty equalities closely mirror the uncertainty relation detailed in \cite{ML14S}
\begin{equation}
\Delta u^2+\Delta v^2\geqslant\norm{\braketd{\psi_k^{\perp}}{u\mp iv}{\psi}}^2\pm i\brakets{[u,v]},
\end{equation}
\begin{equation}
\Delta u\Delta v\geqslant
\frac{\pm\frac{1}{2}i\brakets{[u,v]}}{1-\frac{1}{2}\norm{\braketd{\psi_k^{\perp}}{\frac{u}{\Delta u}\mp i\frac{v}{\Delta v}}{\psi}}^2}.
\end{equation}

Similarly, one can also reproduce the uncertainty relations for Hermitian operators from the unitary uncertainty relations $(\ref{Eq-UURS},\ref{Eq-UURP})$ by expressing the unitary operators in terms of the Hermitian operators and keeping the terms up to second order. These instances show an intrinsic connection between the uncertainty relations of unitary operators and their counterparts in the realm of Hermitian operators.
\section{Conclusions}
\label{Sec-Con}

In this study, we put forth two unitary uncertainty equalities. These equalities are achieved by arbitrary pure state, and give rise to two hierarchical structures of unitary uncertainty relations. We also delve into the extension of uncertainty equalities and inequalities to mixed states. Notably, one of our unitary uncertainty equalities (\ref{Eq-UUES}) precisely reproduces the unitary uncertainty relation (\ref{Eq-BPUUR2}) proposed by Bagchi and Pati. To assess the effectiveness of our findings, we compared our results with the bounds established by Bagchi and Pati, Massar and Spindel, as well as Bong {\it et al}. Specifically, we examined the bounds of various unitary uncertainty relations for two unitary operators related by the discrete Fourier transform, considering a class of states $\ket{\psi}=\cos\theta\ket{0}-\sin\theta\ket{d-1}$ in the $d$-dimensional Hilbert space. Our analysis reveals that some of our unitary uncertainty relations are more stronger than both MSUUR, BPUURs, and BUUR, respectively. Moreover, it unveils a fundamental connection between unitary uncertainty equalities (relations) of unitary operators and their counterparts in the realm of Hermitian operators, particularly for higher dimensional systems and a subset of quantum states. For further investigations, since quantum information scrambling has an intricate relationship with uncertainty principle, which can be measured by an increasing commutator of two unitary operators, the unitary uncertainty equalities and inequalities may be employed to obtain nontrivial bounds for it.
\section*{Acknowledgements}
\noindent
This work was supported in part by National Natural Science Foundation of China(NSFC) under the Grants 11975236 and 12235008, and  University of Chinese Academy of Sciences.
\bibliographystyle{ref_style.bst}
\bibliography{references.bib}
\end{document}